\documentclass{amsart}

\usepackage{amsmath,amssymb,amsthm} 
\usepackage{bbm}
\usepackage{stmaryrd}
\usepackage{graphicx}

\DeclareMathAlphabet\scr{U}{scr}{m}{n}
\SetMathAlphabet\scr{bold}{U}{scr}{b}{n}
  \DeclareFontFamily{U}{scr}{\skewchar\font'177}%
  \DeclareFontShape{U}{scr}{m}{n}{<-6>rsfs5<6-8>rsfs7<8->rsfs10}{}%
  \DeclareFontShape{U}{scr}{b}{n}{<-6>rsfs5<6-8>rsfs7<8->rsfs10}{}%

\newtheorem{theorem}{Theorem}[section]  

\newtheorem{standingassumption}[theorem]{Standing Assumption}

\newtheorem{corollary}[theorem]{Corollary}

\newtheorem{definition}[theorem]{Definition}
\newtheorem{lemma}[theorem]{Lemma}
\newtheorem{proposition}[theorem]{Proposition}

\theoremstyle{definition}

\newtheorem{remark}[theorem]{Remark}

\numberwithin{equation}{section}

\begin{document}

\title{Asymptotics and Duality for the \linebreak Davis and Norman Problem}
\author{Stefan Gerhold}
\address{Technische Universit\"at Wien, Institut f\"ur Wirtschaftsmathematik\endgraf Wiedner Hauptstrasse 8-10, A-1040 Wien, Austria}
\email{sgerhold@fam.tuwien.ac.at}
\author{Johannes Muhle-Karbe}
\address{ETH Z\"urich, Departement f\"ur Mathematik\endgraf R\"amistrasse 101, CH-8092, Z\"urich, Switzerland}
\email{johannes.muhle-karbe@math.ethz.ch}
\author{Walter Schachermayer}
\address{Universit\"at Wien, Fakult\"at f\"ur Mathematik\endgraf Nordbergstrasse 15, A-1090 Wien, Austria}
\email{walter.schachermayer@univie.ac.at}
\date{\today}

\thanks{We thank Paolo Guasoni, Mete Soner, and, in particular, Steve Shreve for valuable discussions and comments. We are also grateful to an anonymous referee for his/her careful reading of the manuscript. The first author was partially supported by the Austrian Federal Financing Agency and the Christian-Doppler-Gesellschaft (CDG). The second author gratefully acknowledges partial support by the National Centre of Competence in Research ``Financial Valuation and Risk Management'' (NCCR FINRISK), Project D1 (Mathematical Methods in Financial Risk Management), of the Swiss National Science Foundation (SNF). The third author was partially supported by the Austrian Science Fund (FWF) under grant P19456, the European Research Council (ERC) under grant FA506041, the Vienna Science and Technology Fund (WWTF) under grant MA09-003, and by the Christian-Doppler-Gesellschaft (CDG).}

\begin{abstract}
We revisit the problem of maximizing expected logarithmic utility from consumption over an infinite horizon in the Black-Scholes model with proportional transaction costs, as studied in the seminal paper of Davis and Norman [\emph{Math.\ Operation Research},\ 15, 1990]. Similarly to Kallsen and Muhle-Karbe [\emph{Ann.\ Appl.\ Probab.},\ 20, 2010], we tackle this problem by determining a \emph{shadow price}, that is, a frictionless price process with values in the bid-ask spread which leads to the same optimization problem. However, we use a different parametrization, which facilitates computation and verification. Moreover, for small transaction costs, we determine fractional Taylor expansions of arbitrary order for the boundaries of the no-trade region and the value function. This extends work of Jane{\v{c}}ek and Shreve [\emph{Finance Stoch.},\ 8, 2004], who determined the leading terms of these power series. 
\end{abstract}

\maketitle

\section{Introduction}
It is a classical problem of financial theory to maximize expected utility from consumption (cf., e.g., \cite{karatzas.shreve.98, merton.90} and the references therein). This is often called the \emph{Merton problem}, because it was first formulated and solved in a continuous-time setting by Merton~\cite{merton.69, merton.71}. More specifically, he found that -- for logarithmic or power utility and one risky asset following geometric Brownian motion -- it is optimal to keep the fraction of wealth invested into stocks equal to a constant $\theta$, which is known explicitly in terms of the model parameters.

Magill and Constantinides \cite{constantinidis.magill.76} extended Merton's setting to incorporate proportional transaction costs. In particular, they showed that -- again for logarithmic or power utility -- it is optimal to engage in the minimal amount of trading necessary to keep the fraction of wealth in stocks inside some \emph{no-trade region} $[\underline{\theta},\overline{\theta}]$ around $\theta$. 

Their somewhat heuristic derivation was made rigorous in the seminal paper of Davis and Norman \cite{davis.norman.90}, who also showed how to compute $\underline{\theta}, \overline{\theta}$ by solving a free boundary problem. 

Using the theory of viscosity solutions, this was further generalized by Shreve and Soner \cite{shreve.soner.94}. Moreover, the same approach allowed Jane{\v{c}}ek and Shreve \cite{janecek.shreve.04} to derive rigorous first-order asymptotic expansions of $\underline{\theta}, \overline{\theta}$ and the value function for small transaction costs (also cf.\ \cite{constantinidis.86, shreve.soner.94, rogers.04,whalley.wilmott.97} for related asymptotic results).

All of the above papers use methods from the theory of stochastic control. On the other hand, dating back to the pioneering papers of Jouini and Kallal \cite{jouini.kallal.95} and Cvitani\'c and Karatzas \cite{cvitanic.karatzas.96}, much of the general theory for markets with transaction costs is formulated in terms of duality theory leading to \emph{consistent price systems} resp.\ \emph{shadow prices}. These are frictionless markets evolving within the bid-ask spread of the original market with transaction costs, which lead to an equivalent dual optimization problem. For logarithmic utility, Kallsen and Muhle-Karbe \cite{kallsen.muhlekarbe.10} solved the Merton problem with transaction costs by simultaneously determining a shadow price and its optimal portfolio, i.e., by determining the dual and the primal optimizer at the same time (also compare \cite{kuehn.stroh.10} for an extension of this approach to a limit order market). 

In the present article, we revisit this problem using a different parametrization in the spirit of \cite{gerhold.al.10}. We employ arguments that are tailor-made for sufficiently small transaction costs, i.e., which focus on asymptotic expansions and thus seem promising for tackling more complicated models. In particular, we determine rigorous asymptotic expansions both for the boundaries $\underline{\theta}, \overline{\theta}$ of the no-trade region and the corresponding value function, where terms of arbitrary order can be algorithmically computed. This extends the first-order expansions of \cite{janecek.shreve.04}. However, we emphasize that we only consider the particularly simple case of logarithmic utility, unlike \cite{janecek.shreve.04}, who also deal with the more involved case of power utilities. An extension of the present approach to power utilities is subject of current research.

The remainder of this article is organized as follows. After introducing the Merton problem with transaction costs in Section 2, we present a heuristic derivation of our candidate shadow price process. Subsequently, in Section 4, we prove that this candidate is indeed well-defined and a shadow price process for sufficiently small transaction costs. Moreover, we derive asymptotic expansions for the boundaries of the no-trade region. Afterwards, we also determine an asymptotic expansion for the value function. Finally, in Section 6, we show that the shadow price corresponds to the solution of the dual problem.

\section{The Merton Problem with transaction costs}\label{s:setup}
We study the problem of maximizing expected logarithmic utility from consumption over an infinite horizon in the presence of proportional transaction costs, as in \cite{davis.norman.90, kallsen.muhlekarbe.10,shreve.soner.94}. Fix a filtered probability space $(\Omega,\scr{F},(\scr{F}_t)_{t\in \mathbb{R}_+},P)$, whose filtration is generated by a standard Brownian motion $(W_t)_{t\in \mathbb{R}_+}$. We consider a market with two investment opportunities, namely a bond and a stock. The \emph{ask price process} $S$ of the stock is supposed to follow the Black-Scholes model
\begin{equation}\label{eq:bs}
dS_t/S_t=\mu dt +\sigma dW_t, \quad S_0>0,\quad \mu,\sigma>0,
\end{equation}
and the \emph{bid price process} is assumed to be given by $(1-\lambda)S$ for some $\lambda \in (0,1)$\footnote{This notation, also used in~\cite{taksar.al.88}, turns out to be convenient in the sequel. It is equivalent to the usual setup with the same constant proportional transaction costs for purchases and sales (compare, e.g.,~\cite{davis.norman.90, janecek.shreve.04, kallsen.muhlekarbe.10, shreve.soner.94}). Indeed, set $\check{S}=\frac{2-\lambda}{2}S$ and $\check{\lambda}=\frac{\lambda}{2-\lambda}$. Then $((1-\lambda)S,S)$ coincides with $((1-\check{\lambda})\check{S},(1+\check{\lambda})\check{S})$. Conversely, any bid-ask process $((1-\check{\lambda})\check{S},(1+\check{\lambda})\check{S})$ with $\check{\lambda} \in (0,1)$ equals $((1-\lambda)S,S)$ for $S=(1+\check{\lambda})\check{S}$ and $\lambda=\frac{2\check{\lambda}}{1+\check{\lambda}}$.}. This means that one has to pay the higher
price $S_t$ when purchasing the stock at time~$t$, but only receives the lower price $(1-\lambda)S_t$ when selling it. The price of the bond is supposed to be constant and equal to $S^0=1$. As in \cite{davis.norman.90, kallsen.muhlekarbe.10}, we make the following standing assumption, which ensures that the holdings in bond and stock remain positive at all times.

\begin{standingassumption}
$$ 0< \theta:= \frac{\mu}{\sigma^2}<1.$$
\end{standingassumption}

The case $\theta>1$ can be dealt with using minor modifications of the present approach (cf.\ \cite{gerhold.al.10}). The degenerate case $\theta=1$ has to be treated separately, see \cite{janecek.shreve.04, shreve.soner.94}.

Since transactions of infinite variation lead to immediate bankruptcy (compare \cite{campi.schachermayer.06}), we confine ourselves to the following set of trading strategies.

\begin{definition}A \emph{trading strategy} is an $\mathbb{R}^2$-valued predictable  process $(\varphi^0,\varphi)$ of finite variation. Here, $(\varphi^0_{0-}, \varphi_{0-})=(\eta_B,\eta_S)\in \mathbb{R}^2_+ \backslash \{0,0\}$ represents the \emph{initial endowment} in bonds and stocks, whereas $\varphi_t^0$ and $\varphi_t$ denote the units held in bond and in stock at time~$t$, respectively. A (discounted) \emph{consumption rate} is an $\mathbb{R}_+$-valued, adapted stochastic process~$\kappa$ satisfying $\int_0^t \kappa_s ds<\infty$ a.s.\ for all $t \geq 0$. A pair $((\varphi^0,\varphi),\kappa)$ of a trading strategy $(\varphi^0,\varphi)$ and a consumption rate $\kappa$ is called \emph{portfolio/consumption process}.
\end{definition}

To capture the notion of a self-financing strategy, we use the intuition that no funds are added or withdrawn. As in \cite{kallsen.muhlekarbe.10}, this leads to the following notion.

\begin{definition}\label{defi:selffinancing}
A portfolio/consumption process $((\varphi^0,\varphi),\kappa)$ is called \emph{self-financing}, if
\begin{equation}\label{e:selff2}
\varphi^0=\varphi^0_{0-}+\int_0^\cdot (1-\lambda) S_{t} d\varphi^{\downarrow}_t - \int_0^\cdot S_{t} d\varphi^{\uparrow}_t-\int_0^\cdot \kappa_t dt,
\end{equation}
where $\varphi=\varphi^{\uparrow}-\varphi^{\downarrow}$ for increasing predictable processes $\varphi^{\uparrow},\varphi^{\downarrow}$ which do not grow at the same time.
\end{definition}

Note that since~$S$ is continuous and~$\varphi$ is of finite variation, integration by parts yields that this definition coincides with the usual notion of self-financing strategies in the absence of transaction costs if we let $\lambda=0$.

The subsequent definition requires the investor to be solvent at all times. For frictionless markets, i.e., if $\lambda=0$, this coincides with the usual notion of non-negativity of wealth.
\begin{definition}\label{def:opti}
A self-financing portfolio/consumption process $((\varphi^0,\varphi),\kappa)$ is called \emph{admissible}, if the \emph{liquidation wealth process} 
$$V_t(\varphi^0,\varphi):=\varphi_t^0+\varphi_t^+(1-\lambda)S_t-\varphi_t^-S_t, \quad t\geq 0, $$
of $(\varphi^0,\varphi)$ is a.s.\ nonnegative. 
\end{definition}

We consider the following classical investment/consumption problem.

\begin{definition}[Primal problem]
Given an initial endowment $(\eta_B,\eta_S)$, an admissible portfolio/consumption process $((\varphi^0,\varphi),\kappa)$ is called \emph{optimal} if it attains the maximum in
\begin{equation}\label{e:DiscountedUtility}
u_{(1-\lambda)S,S}(\eta_B,\eta_S):=\sup_{((\varphi^0,\varphi),\kappa)}E\left(\int_0^{\infty} e^{-\delta t} \log(\kappa_t)dt\right)
\end{equation}
over all admissible portfolio/consumption processes $((\varphi^0,\varphi),\kappa)$. Here $\delta>0$ denotes a fixed given \emph{impatience rate}.
\end{definition}

Note that since $\delta>0$, the value function of the Merton problem \emph{without} transaction costs is finite by \cite[Theorem 2.1]{davis.norman.90}. This upper bound for the value function in the present setup with transaction costs then implies that the latter is finite, since it is also bounded from below by the finite utilities that can be obtained from liquidating at time zero and consuming at a sufficiently small constant relative rate.

The dual problem corresponding to \eqref{e:DiscountedUtility} can be defined as follows (compare \cite{cvitanic.karatzas.96}). 

\begin{definition}[Dual problem]\label{def:dual}
Denote by $\tilde{U}(t,y)=-e^{-\delta t}(\log(e^{\delta t} y)+1)$ the \emph{conjugate function} of the utility function $U(t,x)=e^{-\delta t}\log(x)$, i.e., the Legendre transform of $-U(t,-x)$. Then a pair $(\tilde{S},Z)$ is called \emph{dual optimizer}, if it attains the minimum in
\begin{equation}\label{e:dual}
v_{(1-\lambda)S,S}(\eta_B,\eta_S):=\inf_{(\tilde{S},Z)} E\left(\int_0^\infty \tilde{U}\left(t,\frac{Z_t}{\delta(\eta_B+\eta_S \bar{S}_0)}\right) dt\right)+\frac{1}{\delta}.
\end{equation}
Here $(\tilde{S},Z)$ runs through the set of all continuous semimartingales $\tilde{S}$ with values in the bid-ask spread $[(1-\lambda)S,S]$ and the density processes $Z$ of corresponding martingale measures, i.e., the positive martingales with initial value $1$ for which $\tilde{S}Z$ is a martingale.
\end{definition}

\begin{remark}
The pairs $(\tilde{S},Z)$ correspond to the consistent price systems of \cite{guasoni.al.08,schachermayer.04}. In the frictionless case $\lambda=0$ there is only one of these, namely the stock price process $S$ itself and the corresponding density process $Z$ determined by Girsanov's theorem. In this case, the dual problem reduces to the usual notion for frictionless markets as in, e.g., \cite[Theorems 3.9.11 and 3.9.12]{karatzas.shreve.98}. The only difference is the constant $1/\delta$, which is added to let the optimal primal and dual values coincide (compare Theorem \ref{thm:dual} below).
\end{remark}

\section{Heuristic derivation of the shadow price}
The key to solving both the primal and dual problems above is the following concept:

\begin{definition}
A \emph{shadow price process} $\tilde{S}$ is a continuous semimartingale with values in the bid-ask spread $[(1-\lambda)S,S]$, such that the maximal expected utilities $u_{\tilde{S},\tilde{S}}(\eta_B,\eta_S)$ in the frictionless market with price process $\tilde{S}$ and $u_{(1-\lambda)S,S}(\eta_B,\eta_S)$ in the original market with bid-ask process $((1-\lambda)S,S)$ coincide.
\end{definition}

From an economic point of view, it is obvious that the maximal expected utility for \emph{any} frictionless price process $\bar{S}$ with values in the bid-ask spread $[(1-\lambda)S,S]$ is at least as high as for the original problem with transaction costs, since one can always buy at least as cheaply and sell at least as expensively. To achieve equality of the maximal expected utilities, the optimal number $\varphi_t$ of stocks in the frictionless market with the shadow price process $\tilde{S}$ should increase (resp.\ decrease) only when $\tilde{S}_t=S_t$ (resp.\ $\tilde{S}_t=(1-\lambda)S_t$). Otherwise one could achieve higher utility than in the original market with transaction costs. This portfolio will then also be feasible in the original market with transaction costs and hence be optimal for this market as well.

It is well-known, dating back to~\cite{constantinidis.magill.76}, that there is no trading in the original market with transaction costs as long as the fraction 
$$\pi_t:=\frac{\varphi_t S_t}{\varphi_t^0+\varphi_t S_t}=\frac{1}{1+\varphi^0_t/(\varphi_t S_t)}$$
of total wealth in terms of the ask price $S$ invested in stocks lies inside some interval $[\underline{\theta},\overline{\theta}]$ around the Merton proportion $\theta=\mu/\sigma^2$. Let us reparametrize this conveniently for the computations below (compare \cite{gerhold.al.10}): Set 
\begin{equation}\label{eq:parameter}
c:=\frac{1-\underline{\theta}}{\underline{\theta}}, \quad \overline{s}:=\frac{(1-\underline{\theta})\overline{\theta}}{(1-\overline{\theta})\underline{\theta}}
\end{equation}
and define 
$$m_t:=\varphi^0_t /(c\varphi_t).$$
Then $\pi \in [\underline{\theta},\overline{\theta}]$ is equivalent to $m \in [S/\overline{s},S]$. We now make the \emph{ansatz} that the shadow price $\tilde{S}$ is given by
$$\tilde{S}_t=G(m_t,S_t),$$
for some function $G:\mathbb{R}^2 \to \mathbb{R}$. This is equivalent to assuming that $\tilde{S}_t$ is a function of the current ask price $S_t$ and the current fraction $\pi_t$ of wealth, in terms of the ask price~$S$, invested in stocks. However, using $m_t$ instead of $\pi_t$ turns out to be more convenient in the calculations below.

The function $G$ enjoys a scaling property inherited from the logarithmic utility function and its constant \emph{relative} risk aversion: 
$$G(am_t,aS_t)=a\tilde{S}_t=aG(m_t,S_t), \quad \forall a>0.$$
This is because the optimal number $\varphi^0_t$ of bonds changes to $a\varphi^0_t$ whereas the optimal number $\varphi_t$ of stocks and the dimensionless  constant $c$ remain the same when we rescale from $S_t$ to $aS_t$. Therefore, we can pass to the function
$$g(S_t/m_t):=G(1,S_t/m_t),$$
such that
$$\tilde{S}_t=G(m_t,S_t)=m_t g(S_t/m_t).$$

We now heuristically derive a free boundary problem that determines the function $g$ as well as the constants $c$ and $\bar{s}$. With these quantities at hand, we then explain how to construct the process $m$ and in turn the shadow price $\tilde{S}$.

 If $\tilde{S}$ is a shadow price, it is well-known that for logarithmic utility the optimal consumption rate at time $t \geq 0$ in the frictionless market with price process $\tilde{S}$ is given by
$$\kappa_t=\delta (\varphi^0_t+\varphi_t \tilde{S}_t),$$
where $\delta>0$ denotes the discount factor in the optimization problem (cf., e.g., \cite{merton.90}). Consequently, we  have
\begin{equation}\label{eq:varphi}
d\varphi_t=0, \quad d\varphi^0_t=-\delta(\varphi^0_t+\varphi_t \tilde{S}_t)dt=-\delta \varphi^0_t (1+g(S_t/m_t)/c)dt
\end{equation}
in the no-trade region, where we have inserted our ansatz $\tilde{S}_t=m_t g(S_t/m_t)$ and the definition $m_t=\varphi^0_t/(c\varphi_t)$.  Now consider the process
$$Y_t:=S_t/m_t=S_tc\varphi_t/\varphi^0_t,$$
which represents the fraction of wealth in stocks (in terms of the ask price $S$) over wealth in  bonds, normalized by $c$. In view of the dynamics \eqref{eq:bs} of $S$ and \eqref{eq:varphi} of $\varphi^0$ and $\varphi$, It\^o's formula yields
\begin{equation}\label{eq:y}
dY_t/Y_t=dS_t/S_t-d\varphi^0_t/\varphi^0_t=\big(\mu+\delta(1+g(Y_t)/c)\big)dt+\sigma dW_t.
\end{equation}
Again by It\^o's formula, the dynamics of the shadow price $\tilde{S}=(S/Y)g(Y)$ are
\begin{align*}
d\tilde{S}_t/\tilde{S}_t&=\left(\frac{\mu g'(Y_t)Y_t+\frac{\sigma^2}{2}g''(Y_t)Y_t^2-\delta (g(Y_t)-g'(Y_t)Y_t)(1+g(Y_t)/c)}{g(Y_t)}\right)dt\\
&\qquad +\left(\frac{\sigma g'(Y_t)Y_t}{g(Y_t)}\right)dW_t\\
&=:\tilde{\mu}(Y_t) dt+\tilde{\sigma}(Y_t) dW_t,
\end{align*}
during each excursion into the no-trade region. Due to Merton's rule for logarithmic utility \cite{merton.71}, the Merton ratio for $\tilde{S}$ must equal the proportion of total wealth invested in stocks \emph{in terms of the shadow price $\tilde{S}$}, i.e., we must have
$$\frac{\tilde{\mu}(Y_t)}{\tilde{\sigma}(Y_t)^2}=\tilde{\pi}_t:=\frac{\varphi_t\tilde{S}_t}{\varphi^0_t+\varphi_t \tilde{S}_t}=\frac{g(Y_t)}{c+ g(Y_t)}.$$
After rearranging, this yields the following ODE for $g$:
\begin{equation}\label{eq:ode}
g''(y)=\frac{2g'(y)^2}{c+g(y)}-\frac{2\theta g'(y)}{ y}+\frac{2\delta}{\sigma^2}\left(1+\frac{g(y)}{c}\right)\left(\frac{g(y)}{y^2}-\frac{g'(y)}{y}\right).
\end{equation}
Note that in the limit $\delta \to 0$, this is precisely the ODE from \cite{gerhold.al.10}. By \eqref{eq:parameter}, the process $Y=S/m$ takes values in $[1,\overline{s}]$. Hence it is natural to complement the ODE \eqref{eq:ode} with some boundary conditions at $1$ and $\overline{s}$ as in \cite{gerhold.al.10}. Since we have a second order ODE, and the parameters $c$ and $\overline{s}$ are also yet to be determined, we need two boundary conditions at each endpoint. The shadow price $\tilde{S}=mg(S/m)$ should equal the bid price $(1-\lambda)S$ at the selling boundary $\overline{\theta}$, where $m=S/\overline{s}$, resp.\ the ask price $S$ at the buying boundary $\underline{\theta}$, where $m=S$. Hence we impose
\begin{equation}\label{eq:boundary1}
g(1)=1 \quad \mbox{and} \quad g(\overline{s})=(1-\lambda)\overline{s}.
\end{equation}
Moreover, we add the smooth pasting conditions
\begin{equation}\label{eq:boundary2}
g'(1)=1 \quad \mbox{and} \quad g'(\overline{s})=1-\lambda.
\end{equation}
They can be derived by looking at the ratio $\tilde{S}/S$, which should stay in $[1-\lambda,1]$, and setting its diffusion coefficient equal to zero at $1-\lambda$ and~$1$ (compare \cite{gerhold.al.10}). These four boundary conditions should now determine the function $g$ as the solution to our second-order ODE as well as the two free variables $c$ and $\bar{s}$, which in turn identify the boundaries $\underline{\theta}$ and $\overline{\theta}$ of the no-trade region.

With $g,c,\bar{s}$ at hand, it now remains to extend this construction from one excursion into the no-trade region to the entire positive real time line. To this end notice that, even when trading in the stock does happen, the process $Y=c\varphi_t S_t /\varphi^0_t$ must always remain within $[1,\overline{s}]$, in order to keep the corresponding fraction of wealth in stocks in $[\underline{\theta},\overline{\theta}]$. However, the diffusion coefficient in the SDE \eqref{eq:y} is bounded away from zero when~$Y_t$ approaches either $1$ or $\overline{s}$. Therefore we need to complement the SDE \eqref{eq:y} for $Y$ with \emph{instantaneous reflection} at both endpoints of $[1,\overline{s}]$. Then $Y$ will by definition remain in $[1,\overline{s}]$ and follow the claimed dynamics during each excursion into the no-trade region. With the process $Y$ at hand, we can then define the candidate shadow price $\tilde{S}=(S/Y)g(Y)$.

\section{Construction of the shadow price}

We begin with an existence result for the free boundary value problem from the preceding section. As in \cite[Proposition 4.2]{kallsen.muhlekarbe.10}, it would be possible to obtain a similar result for any relative width $\lambda$ of the bid-ask spread. However, we restrict ourselves to the case of sufficiently small $\lambda$ here. This is all we need to derive asymptotic expansions in terms of $\lambda$ below and seems much better suited to the analysis of more complicated models, since it requires only a much coarser inspection of the involved equations. 

\begin{lemma}\label{le:existence}
If $\lambda>0$ is sufficiently small, then
there are constants $c>(1-\theta)/\theta$ and $\overline{s}>1$ as well as a function $g: [1,\overline{s}] \to [1,(1-\lambda)\overline{s}]$ satisfying the ODE
\begin{equation}\label{eq:ode2}
g''(y)=\frac{2g'(y)^2}{c+g(y)}-\frac{2\theta g'(y)}{ y}+\frac{2\delta}{\sigma^2}\left(1+\frac{g(y)}{c}\right)\left(\frac{g(y)}{y^2}-\frac{g'(y)}{y}\right)
\end{equation}
with boundary conditions
\begin{equation}\label{eq:boundary}
 g(1)=1, \quad g(\overline{s})=(1-\lambda)\overline{s}, \quad g'(1)=1, \quad g'(\bar{s})=(1-\lambda).
 \end{equation}
\end{lemma}
\begin{proof}
  We first ignore the free boundary $\bar{s}$ and consider only the ODE~\eqref{eq:ode2}
  with initial conditions $g(1)=g'(1)=1$, with~$c$ as free parameter.
  This problem can be written as
  \begin{align*}
    \partial_s g(s,c) &= F_1(s,g,h,c), \\
    \partial_s h(s,c) &= F_2(s,g,h,c), \\
    g(1,c) &= h(1,c) = 1,
  \end{align*}
  where $F_1(s,u,v,c) := v$, and
  \[
    F_2(s,u,v,c) :=  \frac{2v^2}{c+u} - \frac{2\theta v}{s}
      +\frac{2\delta}{\sigma^2}\left(1+\frac{u}{c}\right) \left(\frac{u}{s^2} - \frac{v}{s} \right).
  \]
  Since $\bar{c}:=(1-\theta)/\theta\notin\{-1,0\}$, the function~$F_2$ (and of course also~$F_1$) is
  analytic in a complex neighbourhood of~$(s,u,v,c)=(1,1,1,\bar{c})$.
  By a standard existence and
  uniqueness result for analytic ODEs~\cite[Theorem~1.1]{IlYa08},
  we obtain a neighbourhood of $(s,c)=(1,\bar{c})$ with a unique analytic solution $g(s,c)$.
  The coefficients~$a_{ij}$ of its Taylor expansion
  \[
    g(s,c) = 1 + (s-1) + \sum_{i\geq2} \sum_{j\geq0} a_{ij}(s-1)^i (c-\bar{c})^j
  \]
  can be calculated recursively. Note that 
  $a_{0j}=a_{1j}=\delta_{0j}$ for $j\geq0$, where $\delta$ is the Kronecker delta,
  which follows from the initial conditions $g(1,c)=g'(1,c)=1$.
  
  From this analytic function of two variables we will now construct the solution of the
  free boundary value problem \eqref{eq:ode2}--\eqref{eq:boundary}, by supplying appropriate
  quantities~$c$ and~$\bar{s}$
  (depending on~$\lambda$). These should satisfy
  \begin{equation}\label{eq:fb}
    \lambda \bar{s} = \bar{s}-g(\bar{s},c) \qquad \text{and}
      \qquad g(\bar{s},c) - \bar{s} \partial_s|_{s=\bar{s}} g(s,c) = 0.
  \end{equation}
  We divide the latter equation by $s-1$, reflecting the fact that we are not
  interested in the solution $s=1$ (which exists because of the initial conditions $g(1)=g'(1)=1$).
  A Taylor expansion yields
  \begin{equation}\label{eq:divided}
    \frac{g(s,c) - s \partial_s g(s,c)}{s-1} =
      \sum_{i\geq0} \sum_{j\geq0} b_{ij}(s-1)^i (c-\bar{c})^j
  \end{equation}
  for some real, explicitly computable coefficients~$b_{ij}$. It turns out that
  $b_{0,0}=0$, while
  $b_{1,0}=\theta(\theta-1)$ is non-zero, so that we can apply the implicit function theorem
  for multivariate analytic functions~\cite[Theorem~I.B.4]{GuRo65} to obtain an analytic function
  \begin{equation}\label{eq:sbar of c}
    H(c) = 1+\sum_{i\geq1} d_i (c-\bar{c})^i
      = 1+\tfrac{2\theta}{1-\theta}(c-\bar{c}) + O((c-\bar{c})^2)
  \end{equation}
  such that~\eqref{eq:divided} vanishes if~$H(c)$ is substituted for~$s$.
  Here and in what follows, the coefficients of the power series symbolized
  by the $O(\cdot)$-s can be algorithmically computed.
  Now we insert
  the function~$H(c)$ into the right hand side of the first equation in~\eqref{eq:fb}:
  \[
    H(c) - g(H(c),c) = a_0^{-1} (c-\bar{c})^3 + O((c-\bar{c})^4),
  \]
  where $a_0=\tfrac34 (1-\theta)^2/\theta^4$. If we insert~$H(c)$ also on the left hand side,
  the first equation in~\eqref{eq:fb} becomes
  \[
    \lambda(1+d_1(c-\bar{c}) + O((c-\bar{c})^2)) = a_0^{-1}(c-\bar{c})^3(1+O(c-\bar{c})).
  \]
  Dividing by $ a_0^{-1}(1+O(c-\bar{c}))$, we obtain
  \[
    \lambda  (a_0 + O(c-\bar{c})) = (c-\bar{c})^3.
  \]
  Let us denote the coefficients of the series on the left hand side by~$a_k$:
  \[
    \lambda (a_0 + a_1(c-\bar{c}) + \dots ) = (c-\bar{c})^3.
  \]
  Now take the third root:
  \[
    \lambda^{1/3}a_0^{1/3} (1 + \tfrac{a_1}{3a_0}(c-\bar{c})+  \dots) = c-\bar{c}.
  \]
  (Note that the coefficients of powers of power series can be algorithmically
  calculated~\cite{Go74}.)
  Since $a_0=\tfrac34 (1-\theta)^2/\theta^4$ is non-zero, we can apply the implicit
  function theorem, so that the latter equation has a unique solution
  for small~$\lambda$. It is an analytic function of~$\lambda^{1/3}$:
  \[
    c - \bar{c} = a_0^{1/3} \lambda^{1/3} + O(\lambda^{2/3}).
  \]
  The coefficients in this power series expansion can be calculated with
  the Lagrange inversion formula:
  \begin{equation}\label{eq:lag inv}
    [\lambda^{k/3}](c-\bar{c}) = \tfrac{1}{k}[z^{k-1}]a_0^{k/3}
      (1 + \tfrac{a_1}{3a_0} z + \dots)^k ,\qquad k\geq1,
  \end{equation}
  where the operator~$[z^k]$ extracts the $k$-th coefficient of a power series.
  (This is analogous to the proof of Proposition~6.1 in~\cite{gerhold.al.10}.)
  With this value of~$c$, the function $g(s,c)$, together with $\bar{s}=H(c)$,
  is the desired solution of the free boundary value problem. To see that
  it maps $[1,\overline{s}]$ to $[1,(1-\lambda)\overline{s}]$, note that otherwise
  the equation $g(s)-sg'(s)=0$ would have a solution in the open interval $]1,\overline{s}[$.
  This is impossible, since~$H(c)$ is the \emph{unique} zero of~\eqref{eq:divided}.
  
  {}From~\eqref{eq:lag inv}, we can calculate as many terms in the expansion
  of~$c$ as we want. The first of them are
  \begin{equation}\label{eq:c expans}
    c = \bar{c} + \frac{1-\theta}{2\theta}\left( \frac{6}{\theta(1-\theta)} \right)^{1/3} \lambda^{1/3} 
    +\frac{2\delta/\sigma^2+3(1-\theta)^2}{12\theta} \left( \frac{6}{\theta(1-\theta)} \right)^{2/3}   
    \lambda^{2/3} + O(\lambda).
  \end{equation}
  Plugging this expansion into~\eqref{eq:sbar of c} gives an expansion of~$\bar{s}$,
  again to arbitrary order. The first terms are
  \begin{equation}\label{eq:s expans} 
    \bar{s} = 1 + \left( \frac{6}{\theta(1-\theta)} \right)^{1/3} \lambda^{1/3} +\frac12 \left(   
    \frac{6}{\theta(1-\theta)} \right)^{2/3} \lambda^{2/3} + O(\lambda).   
  \end{equation}
  Finally, note that~\eqref{eq:c expans} shows that $c>\bar{c}=(1-\theta)/\theta$ for sufficiently   
  small~$\lambda$.
\end{proof}

Henceforth, we consider sufficiently small transaction costs $\lambda>0$ for Lemma \ref{le:existence} to be valid and write $g$ resp.\ $c, \bar s$ for the corresponding function resp.\ constants defined therein.  

The existence of the process $Y$ as the solution to the SDE \eqref{eq:y} with reflection at $1$ and $\overline{s}$ is a consequence of a classical result of Skorokhod \cite{skorokhod.61}.

\begin{lemma}\label{lem:skorokhod}
Let $y \in [1,\bar{s}]$. Then there exists a strong solution to the \emph{Skorokhod SDE}
\begin{equation}\label{eq:skorokhod}
dY_t=Y_t\big(\mu+\delta(1+g(Y_t)/c)\big)dt+Y_t\sigma dW_t, \quad Y_0=y,
\end{equation}
with instantaneous reflection at $1$ and $\overline{s}$, that is, a continuous, adapted, $[1,\overline{s}]$-valued process $Y$ and nondecreasing adapted processes $L$ and $U$ increasing only on the sets $\{Y_t=1\}$ and $\{Y_t=\overline{s}\}$, respectively, such that
$$Y_t=Y_0+\int_0^t Y_u (\mu+\delta(1+g(Y_u)/c)du+\int_0^t Y_u\sigma dW_u+\int_0^t Y_u dL_u-\int_0^t Y_u dU_u$$
holds for all $t \in \mathbb{R}_+$.
\end{lemma}

\begin{remark}
In the limit $\delta \to 0$, the process~$Y$ is just a geometric Brownian motion instantaneously reflected at $1$ and $\overline{s}$. Hence it coincides with the process $S/m$ from~\cite{gerhold.al.10}. Let us also point out that $Y$ is related to the process $\beta$ from \cite[Lemma 4.3]{kallsen.muhlekarbe.10} by a nonlinear transformation. However, the present parametrization evidently leads to simpler dynamics. Moreover, the additional process $C$ from \cite[Lemma 4.4]{kallsen.muhlekarbe.10} is no longer necessary.
\end{remark}

\begin{proof}[Proof of Lemma \ref{lem:skorokhod}]
Since $g$ is bounded on $[1,\overline{s}]$ and $c>(1-\theta)/\theta>0$ by Lemma~\ref{le:existence}, the coefficients of \eqref{eq:skorokhod} are globally Lipschitz on $[1,\overline{s}]$. Hence the claim follows from~\cite{skorokhod.61}.
\end{proof}

It was shown in \cite[Proposition 2.1]{gerhold.al.10} that, for $\delta=0$, the process $\tilde{S}=m g(S/m)=(S/Y) g(Y)$ is an It\^o process, i.e., does not involve local time. We now deduce from Lemma \ref{lem:skorokhod} that this still holds for arbitrary $\delta > 0$. The present proof only uses It\^o's formula and the smooth pasting conditions~\eqref{eq:boundary} for~$g$. In particular, it allows to avoid the somewhat ad hoc localization procedures used in the proof of \cite[Proposition 2.1]{gerhold.al.10}.

\begin{proposition}\label{prop:shadow}
For $y \in [1,\bar{s}]$ and $Y$ as in Lemma \ref{lem:skorokhod}, $\tilde{S}:=(S/Y) g(Y)$ is a positive It\^o process of the exponential form
$$\tilde{S}=\frac{S_0 g(y)}{y} \scr{E}(\tilde{X}),$$
with
$$d\tilde{X}_t:=d\tilde{S}_t/\tilde{S}_t= \left(\frac{\sigma^2 g'(Y_t)^2 Y_t^2}{(c+g(Y_t))g(Y_t)}\right)dt+\left(\frac{g'(Y_t)Y_t \sigma}{g(Y_t)}\right)dW_t, \quad \tilde{X}_0=0,$$
and $\tilde{S}$ takes values in the bid-ask spread $[(1-\lambda)S,S]$.
\end{proposition}

\begin{proof}
Using It\^o's formula, we obtain
$$dg(Y_t)=g'(Y_t)Y_t(\mu+\delta(1+g(Y_t)/c)dt+\sigma dW_t+dL_t - dU_t)+\frac{\sigma^2}{2}g''(Y_t)Y_t^2 dt$$
and
$$\frac{d(S_t/Y_t)}{S_t/Y_t}=-\delta\left(1+\frac{g(Y_t)}{c}\right)dt-dL_t+dU_t.$$
Hence integration by parts yields
\begin{align*}
d\tilde{S}_t/\tilde{S}_t&=\left(\frac{\mu g'(Y_t)Y_t+\frac{\sigma^2}{2}g''(Y_t)Y_t^2-\delta(g(Y_t)-g'(Y_t)Y_t)(1+g(Y_t)/c)}{g(Y_t)}\right)dt\\
&\qquad+\left(\frac{\sigma g'(Y_t)Y_t}{g(Y_t)}\right)dW_t+\frac{g'(Y_t)Y_t-g(Y_t)}{g(Y_t)}(dL_t-dU_t).
\end{align*}
Since $L$ and $U$ only increase on $\{Y_t=1\}$ resp.\ $\{Y_t=\overline{s}\}$, where the integrand $g'(Y_t)Y_t-g(Y_t)$ is zero by the smooth pasting conditions for $g$, the last term vanishes. The It\^o process decomposition of $\tilde{X}$ now follows by inserting the ODE \eqref{eq:ode2} for $g''$. 

For the second part of the assertion, notice that $g''(1)<0$, because $c>(1-\theta)/\theta$. Hence it follows from a comparison argument that $g'(y)y-g(y) \leq 0$ on $[1,\overline{s}]$. As this implies that the derivative $(g'(y)y-g(y))/y^2$ of $g(y)/y$ is non-positive, $g(1)/1=1$ and $g(\overline{s})/\overline{s}=1-\lambda$ yield that $\tilde{S}=Sg(Y)/Y$ is indeed $[(1-\lambda)S,S]$-valued.
\end{proof}

Using standard results for frictionless markets, we can now determine the optimal portfolio/consumption pair for $\tilde{S}=(S/Y)g(Y)$.

\begin{lemma}\label{thm:shadow}
Set
\begin{equation}\label{eq:jump}
y=\begin{cases} 1, &\mbox{if } \eta_B> S_0 c \eta_S,\\  \bar{s}, &\mbox{if } \eta_B< S_0 c \eta_S/\bar{s},\\ S_0 c \eta_S/\eta_B &\mbox{otherwise.}  \end{cases}
\end{equation}
For $Y$ and $\tilde{S}$ as in Lemma \ref{lem:skorokhod} and Proposition \ref{prop:shadow}, respectively, define
$$\tilde{\pi}_t=\frac{1}{1+c/g(Y_t)},$$
as well as
\begin{align*}
\tilde{V}_t&=(\eta_B+\eta_S \tilde{S}_0) \scr{E}\left( \int_0^\cdot \frac{\tilde{\pi}_t}{\tilde{S}_t} d\tilde{S}_t-\int_0^\cdot \delta dt\right), \\
\kappa_t&=\delta \tilde{V}_t,\\
\varphi_t&=\tilde{\pi}_t \tilde{V}_t/\tilde{S}_t, \qquad \varphi^0_t=(1-\tilde{\pi}_t)\tilde{V}_t.
\end{align*}
Then
\begin{align}\label{eq:phi}
\varphi_t&=\varphi_{0}+\int_0^t \frac{\varphi_u c g'(Y_u)Y_u}{g(Y_u)(c+g(Y_u))}dL_u-\int_0^t \frac{\varphi_u c g'(Y_u)Y_u}{g(Y_u)(c+g(Y_u))}dU_u,
\end{align}
and $((\varphi^0,\varphi),\kappa)$ is an optimal portfolio/consumption pair for the frictionless market with price process $\tilde{S}$. The corresponding wealth process and the optimal fraction of wealth invested in stocks in terms of $\tilde{S}$ are given by $\tilde{V}$ and $\tilde{\pi}$, respectively.
\end{lemma}

Note that the first resp.\ second cases in \eqref{eq:jump} occur if the initial fraction of wealth in stock $\eta_S S_0/(\eta_B+\eta_S S_0)$ lies below $\underline{\theta}=(1+c)^{-1}$ resp.\ above $\overline{\theta}=(1+c/\bar{s})^{-1}$. In this case, there is a jump from the initial position $(\varphi^0_{0-},\varphi_{0-})=(\eta_B,\eta_S)$, which moves the initial fraction of wealth to the nearest boundary of the interval $[\underline{\theta},\overline{\theta}]$.  Since this initial bulk trade involves the puchase resp.\ sale of stocks, the initial value of $\tilde{S}$ is chosen so as to match the initial ask resp.\ bid price in this case.

\begin{proof}
Since $d\tilde{S}/\tilde{S}=d\tilde{X}_t:=\tilde{\mu}(Y_t)dt+\tilde{\sigma}(Y_t)dW_t$ is an It\^o process with bounded coefficients and $\tilde{\sigma}$ is bounded away from zero, its optimal portfolio/consumption pair is characterized by standard results for frictionless markets, cf., e.g., \cite[Theorem 3.1]{kallsen.goll.00}. In particular, the optimal fraction of wealth (in terms of $\tilde{S}$) invested in stocks is given by Merton's rule, i.e., equals
$$\frac{\tilde{\mu}(Y_t)}{\tilde{\sigma}(Y_t)^2}=\frac{g(Y_t)}{c+g(Y_t)}=\tilde{\pi}_t$$
as claimed. Moreover, the optimal wealth process, the optimal consumption rate and the optimal numbers of bonds resp.\ stocks are indeed given by $\tilde{V}$, $\kappa$, $\varphi^0$, and~$\varphi$. Now It\^o's formula and the ODE \eqref{eq:ode2} for $g$ imply
$$d\tilde{\pi}_t=\left(\frac{\delta g(Y_t)}{c+g(Y_t)}\right)dt+\left(\frac{c \sigma g'(Y_t)Y_t}{(c+g(Y_t))^2}\right)dW_t+\left(\frac{c g'(Y_t) Y_t}{(c+g(Y_t))^2}\right)(dL_t-dU_t).$$
Moreover, by Yor's formula, we have
\begin{align*}
&\frac{\scr{E}(\int_0^\cdot \frac{g(Y_t)}{c+g(Y_t)}d\tilde{X}_t-\int_0^\cdot \delta dt)}{\scr{E}(\tilde{X})}\\
&\qquad=\scr{E}\left(\int_0^\cdot\frac{c}{c+g(Y_t)}d(\langle \tilde{X},\tilde{X} \rangle_t-\tilde{X}_t)-\int_0^\cdot \delta dt\right)\\
&\qquad=\scr{E}\left(\int_0^\cdot \frac{c^2\sigma^2 g'(Y_t)^2 Y_t^2}{g(Y_t)^2(c+g(Y_t))^2}dt-\int_0^\cdot \delta dt-\int_0^\cdot \frac{c \sigma g'(Y_t) Y_t}{g(Y_t)(c+g(Y_t))}dW_t\right).
\end{align*}
Taking into account that the process $L-U$ is continuous and of finite variation, integration by parts finally yields the claimed representation for $\varphi$. 
\end{proof}

The representation \eqref{eq:phi} and the definition of $\tilde{S}_0$ via \eqref{eq:jump} imply that stocks are only purchased (resp.\ sold) when $\tilde{S}=S$ (resp.\ $\tilde{S}=(1-\lambda)S$). Hence the optimal portfolio/consumption process $((\varphi^0,\varphi),\kappa)$ for $\tilde{S}$ is also admissible for the bid-ask process $((1-\lambda)S,S)$. Since shares can be bought at least as cheaply and sold as least as expensively when trading $\tilde{S}$ instead of $(1-\lambda)S,S)$, it follows that the maximal expected utilities coincide, i.e., $\tilde{S}$ is a shadow price. Made precise, this is the content of the following analogue of \cite[Theorem 4.6]{kallsen.muhlekarbe.10}.

\begin{theorem}\label{thm:shadow2}
The portfolio/consumption process $((\varphi^0,\varphi),\kappa)$ from Lemma \ref{thm:shadow} is also optimal for the bid-ask process $((1-\lambda)S,S)$. In particular, $\tilde{S}$ is a shadow price in this market.
\end{theorem}

\begin{proof}
This follows verbatim as in the proof of \cite[Theorem 4.6]{kallsen.muhlekarbe.10}.
\end{proof}

In view of Lemma \ref{thm:shadow} and the definition of $g$, the optimal proportion $\tilde{\pi}$ of stocks in terms of the shadow price $\tilde{S}$ takes values in the interval $[(1+c)^{-1},(1+c/((1-\lambda)\overline{s})^{-1}]$. Translating this into terms of the ask price $S$, we find that the optimal fraction $\pi$ in terms of $S$ is kept within $[(1+c)^{-1},(1+c/\overline{s})^{-1}]$, i.e., the lower resp.\ upper boundaries of the no-trade region are given by $\underline{\theta}=(1+c)^{-1}$ resp.\ $\overline{\theta}=(1+c/\overline{s})^{-1}$.
The following series expansions are immediate consequences of~\eqref{eq:c expans}
and~\eqref{eq:s expans} (upon taking one additional term in both expansions).

\begin{corollary}\label{cor:notradeasymp}
The lower and upper boundaries of the no-trade region in terms of the ask price $S$ have the expansions
  \begin{align*}
    \underline{\theta}&=\frac{1}{1+c} = \theta - \left(\frac34 \theta^2(1-\theta)^2\right)^{1/3} \lambda^{1/3}
    - \frac{\delta \theta}{6\sigma^2}\left( \frac{6}{\theta(1-\theta)} \right)^{2/3} \lambda^{2/3} \\
    & \qquad + \frac{\theta(9(1-\theta)^2(1-2\theta+2\theta^2)-4\delta^2/\sigma^4+6\delta/\sigma^2(3-8\theta+5\theta^2))}{360(1-\theta)} \frac{6}{\theta(1-\theta)}\lambda \\
    &\qquad + O(\lambda^{4/3})
  \end{align*}
  and
  \begin{align*}
   \overline{\theta}&= \frac{1}{1+c/\bar{s}} = \theta + \left(\frac34 \theta^2(1-\theta)^2\right)^{1/3} \lambda^{1/3}
    - \frac{\delta \theta}{6\sigma^2}\left( \frac{6}{\theta(1-\theta)} \right)^{2/3} \lambda^{2/3} \\
    &\qquad - \frac{\theta(3(1-\theta)^2(3-26\theta+26\theta^2)-4\delta^2/\sigma^4+6\delta/\sigma^2(3-8\theta+5\theta^2))}{360(1-\theta)} \frac{6}{\theta(1-\theta)}\lambda  \\
    &\qquad + O(\lambda^{4/3}),
  \end{align*}
  respectively. Consequently, the size of the no-trade region in terms of $S$ satisfies
  \begin{align*}
    \overline{\theta} - & \underline{\theta} = \frac{1}{1+c/\bar{s}} - \frac{1}{1+c} \\
    &= (6\theta^2(1-\theta)^2)^{1/3}\lambda^{1/3} \\
    & - \frac{\theta(3(1-\theta)^2(3-16\theta+16\theta^2)-4\delta^2/\sigma^4+6\delta/\sigma^2(3-8\theta+5\theta^2))}{180(1-\theta)} \frac{6}{\theta(1-\theta)}\lambda \\
     & + O(\lambda^{4/3}).
  \end{align*}
\end{corollary}

Further terms of these expansions can be algorithmically computed, if desired.

The first correction terms are of order $\lambda^{1/3}$, which has been conjectured in \cite[Remark B.3]{shreve.soner.94} and proved in \cite{janecek.shreve.04}. Whereas the first-order corrections are symmetric around the Merton proportion $\theta$, the second correction terms of order $\lambda^{2/3}$ are both negative. This means that it is optimal to start transacting for smaller fractions of stock than indicated by the first-order expansions. As pointed out by \cite{janecek.shreve.04}, an intuitive explanation is that consumption from the bank account continuously increases the proportion of stock. Accordingly, the $\lambda^{2/3}$-terms vanish in the limit for $\delta \to 0$. Interestingly, however, the impatience rate $\delta>0$ only shows up starting from the third term of the expansion for the \emph{width} of the no-trade region even in the presence of consumption.

\begin{remark}
We now compare the expansions from Corollary~\ref{cor:notradeasymp} to the results of Jane{\v{c}}ek and Shreve~\cite{janecek.shreve.04} mentioned above.
We rename their parameter~$\lambda$ as~$\check{\lambda}$; it is related to our~$\lambda$
by $\check{\lambda}=\lambda/(2-\lambda)$ (see the footnote in Section~\ref{s:setup}).
The bounds of the no-trade
region become
\begin{equation}\label{eq:js lower}
  \frac{1}{1+2c/(2-\check{\lambda})} = \theta-\left(\frac32 \theta^2(1-\theta)^2\right)^{1/3} \check{\lambda}^{1/3}
  - \frac{2\theta\delta}{\sigma^2}(12\theta^2(1-\theta)^2)^{-1/3} \check{\lambda}^{2/3} 
  + O(\check{\lambda})
\end{equation}
and
\begin{equation}\label{eq:js upper}
  \frac{1}{1+2c/((2-\check{\lambda})\bar{s})} = \theta +
    \left(\frac32 \theta^2(1-\theta)^2\right)^{1/3} \check{\lambda}^{1/3}
    - \frac{2\theta\delta}{\sigma^2}(12\theta^2(1-\theta)^2)^{-1/3} \check{\lambda}^{2/3}
     + O(\check{\lambda}).
\end{equation}
The $\check{\lambda}^{1/3}$-terms agree with those established in~\cite{janecek.shreve.04},
while our $\check{\lambda}^{2/3}$-terms are not those that were conjectured
in~\cite{janecek.shreve.04}, based on heuristic considerations.

We are indebted to Steve Shreve, who kindly checked these calculations again. While the arguments worked out in \cite{janecek.shreve.04} turned out to be perfectly correct, he found that an assumption on which these calculations were based did not hold true. It was stated in equation (4.27) of \cite{janecek.shreve.04} that ``there is considerable 
evidence'' that the coefficient of the $\lambda$ term in the expansion of  the value function has a certain nonzero value, while subsequent calculations by Steve Shreve indicate that this coefficient is zero.  When this term is set to zero, then heuristic calculations based on the approach of \cite{janecek.shreve.04} lead to \eqref{eq:js lower} and \eqref{eq:js upper} as well.
\end{remark}

\section{Characterization and asymptotic expansion of the value function}

In the frictionless market with price process $\tilde{S}$, standard control methods allow to characterize the value function for the primal problem from Definition \ref{def:opti} as the unique classical solution to the corresponding HJB equation. This yields the maximal utility in the frictionless shadow market, which coincides with the maximal utility in the original market with bid-ask process $((1-\lambda)S,S)$ by Theorem \ref{thm:shadow2}.

\begin{theorem}\label{thm:u ODE}
Let $y$, $Y$, and $\tilde{S}$ be defined as in Lemma~\ref{thm:shadow}. Then, if $\lambda>0$ is sufficiently small, we have
\begin{equation}\label{eq:u repr}
  u_{(1-\lambda)S,S}(\eta_B,\eta_S)=\log(\eta_B+\eta_S \tilde{S}_0)/\delta + w(y),
\end{equation}
where $w$ is the unique (classical) solution to the ODE
\begin{equation}\label{eq:w ODE}
  \frac{\sigma^2}{2} s^2 w''(s)+\left(\mu+\delta\left(1+\frac{g(s)}{c}\right)\right)s w'(s)-\delta    
  w(s)+\frac{\tilde{\mu}^2(s)}{2\delta\tilde{\sigma}^2(s)}+\log(\delta)-1=0,
\end{equation}
on $[1,\bar{s}]$ with Neumann boundary conditions
\[
  w'(1)=w'(\bar{s})=0.
\]
\end{theorem}
\begin{proof}

  We first show that the boundary value problem~\eqref{eq:w ODE} has a unique solution,
  which is an analytic function.
  This is similar to the first part of the proof of Lemma~\ref{le:existence}.
  We first solve the ODE~\eqref{eq:w ODE} with initial conditions $w(1)=w_1+\xi$
  and $w'(1)=0$, with arbitrary~$\xi$, and where
  \[
    w_1 = \frac{\mu^2-2\delta\sigma^2}{2\delta^2\sigma^2} + \frac{\log \delta}{\delta}
  \]  
  is the value to which $w(1)$ converges for $\lambda\to0$.
  This yields an analytic function
  \[
    w(s,c,\xi) = w_1 + \xi + \sum_{i\geq2} \sum_{j,k\geq 0} a_{ijk} (s-1)^i \xi^j (c-\bar{c})^k,
  \]
  for certain real coefficients~$a_{ijk}$.
  (Recall that~$c$ tends to $\bar{c}=(1-\theta)/\theta$ as $\lambda\to0$.)
  To satisfy the boundary condition $w'(\bar{s})=0$,
  we solve the equation
  \[
    \left. \frac{\partial}{\partial s} w(s,c,\xi) \right|_{s=\bar{s}(c)} = 0
  \]
  for~$\xi$.
  After inserting the expansion~\eqref{eq:sbar of c} for~$\bar{s}$, we obtain
  a solution $\xi(c)$, analytic at $c=\bar{c}$. Then $w(s)=w(s,\xi(c(\lambda)), c(\lambda))$
  solves the boundary value problem.
  
  Now fix $\lambda>0$ sufficiently small, such that Lemma~\ref{le:existence},
  Corollary~\ref{cor:notradeasymp}, and the above existence result for~$w$ can be applied. Then, in particular, the mapping $y \mapsto w(y)$ can be extended to a $C^2$-mapping on the whole real line, e.g., by attaching suitable parabolas at $y=1$ and $y=\bar{s}$. It\^o's formula in turn shows that
\begin{align*}
M_t:=&\int_0^t e^{-\delta u} \log(\kappa_u)du+e^{-\delta t}\bigg(\frac{\log(\tilde{V}_t)}{\delta}+w(Y_t)\bigg)\\
=&M_0+\int_0^t e^{-\delta u} \left(\frac{\tilde{\sigma}(Y_u) \tilde{\pi}_u}{\delta}+\sigma Y_u w'(Y_u)\right)dW_u,
\end{align*}
where we have inserted the dynamics of $\tilde{V}$ (cf.\ Lemma 4.5) and $Y$ (cf.\ Lemma 4.2) as well as the ODE and boundary conditions for $w$. In particular, the process $(M_t)_{t \geq 0}$ is a local martingale. By the proof of Proposition \ref{prop:shadow}, we have  $\tilde{\sigma}(y) = \sigma g'(y)/(g(y)/y) \leq \sigma/(1-\lambda)$ for $y \in [1,\bar{s}]$, such that $0 \leq \tilde{\pi} \leq (1+c/((1-\lambda)\bar{s}))^{-1} \leq (1+c/\bar{s})^{-1}$ and $1 \leq Y \leq \bar{s}$ imply
\begin{align*}
&\int_0^t  E\left[e^{-2\delta u} \left(\frac{\tilde{\sigma}(Y_u) \tilde{\pi}_u}{\delta}+\sigma Y_u w'(Y_u)\right)^2\right]du \leq \frac{1}{2\delta}  \left(\frac{\sigma/(1-\lambda)}{\delta(1+c/\bar{s})}+\sigma \bar{s}|w'|_{\max}\right)^2,
\end{align*}   
for $|w'|_{\max}:=\max_{y \in [1,\bar{s}]}|w'(y)|$. Therefore the process $(M_t)_{t \geq 0}$ is in fact a true martingale and also uniformly integrable, since it is bounded in $L^2(P)$ by the It\^o isometry.  Hence $M_t$ converges in $L^1(P)$ to a random variable $M_{\infty}$ with
\begin{equation}\label{eq:convergence}
E(M_{\infty})=E(M_0)=\log(\tilde{V}_0)/\delta+w(Y_0)=\log(\eta_B+\eta_S \tilde{S}_0)/\delta+w(y).
\end{equation}
Now notice that  $e^{-\delta t}w(Y_t) \to 0$, because $w(\cdot)$ is bounded on $[1,\bar{s}]$. Moreover, as $\log(\tilde{V}_t)$ is an It\^o process with bounded coefficients, it follows that $e^{-\delta t}\log(\tilde{V}_t)/\delta \to 0$ in $L^2(P)$ and hence also in probability. Therefore $M_t \to \int_0^\infty e^{-\delta t} \log(\kappa_t) dt$ in probability, such that the a.s.\ uniqueness of limits in probability yields
$$E\left(\int_0^\infty e^{-\delta t} \log(\kappa_t) dt\right)=E(M_{\infty}).$$
Combined with \eqref{eq:convergence}, this completes the proof.
\end{proof}

It does not seem possible to come up with an explicit solution to the above ODE even in terms of the function $g$, unlike in the case without consumption treated in \cite{dumas.luciano.91,taksar.al.88}. However, we can again derive an expansion in powers of $\lambda^{1/3}$.
Note that this expansion cannot be recovered by letting the relative risk aversion
tend to~$1$ in the corresponding result for power utility~\cite{janecek.shreve.04}.
The reason is that the value function for power utility does not converge to the value function
for log utility.

\begin{corollary}
  In addition to the assumptions of Theorem~\ref{thm:u ODE}, suppose that our portfolio
  starts on the Merton line, i.e., 
  $$\pi_0= \frac{\eta_S S_0}{\eta_B+\eta_S S_0}=\theta.$$
  Then the value function satisfies
  \begin{align}
    u_{(1-\lambda)S,S}(\eta_B,\eta_S)=& \frac{1}{\delta} \log(\delta(\eta_B+\eta_S S_0))    
    +\frac{1}{\delta^2}\left(\frac{\mu^2}{2\sigma^2}-\delta\right) \notag \\
    &- \frac{3^{2/3}(\mu(\sigma^2-\mu))^{4/3}}{4\cdot 2^{1/3} \delta^2 \sigma^{10/3}}\lambda^{2/3}
    +O(\lambda). \label{eq:u expans}
  \end{align}

  This expansion can be continued to any order, with algorithmically computable
  coefficients.
\end{corollary}

\begin{proof}
  We expand the first term in~\eqref{eq:u repr}:
  \begin{align*}
    \frac{1}{\delta} \log(\eta_B+\eta_S \tilde{S}_0) &= \frac{1}{\delta}\log(\eta_B+\eta_S S_0 g(y)/y)  \\
    &= \frac{1}{\delta}\log(\eta_B+\eta_S S_0)
      + \frac{1}{\delta}\log\left(1+\theta \left(\frac{g(y)}{y}-1\right)\right)  \\
    &= \frac{1}{\delta}\log(\eta_B+\eta_S S_0) + O(\lambda).
  \end{align*}
  The last line follows after expanding $\log$ into a Taylor series, and
  inserting $y=c/\bar{c}$ and the series for~$g$ and~$c$ from the proof
  of Lemma~\ref{le:existence}.
  As for the second term in~\eqref{eq:u repr}, we use the series for~$w$ found in the proof
  of Theorem~\ref{thm:u ODE} and obtain
  \[
    w(y) = \frac{1}{\delta^2}\left(\frac{\mu^2}{2\sigma^2}-\delta\right) + \frac{\log \delta}{\delta}
    -  \frac{3^{2/3}(\mu(\sigma^2-\mu))^{4/3}}{4\cdot 2^{1/3} \delta^2 \sigma^{10/3}}\lambda^{2/3}
    +O(\lambda).
  \]
  Adding both expansions yields~\eqref{eq:u expans}.
\end{proof}

Note that the expansion~\eqref{eq:u expans} can easily be refined to a bivariate expansion
in terms of~$\lambda^{1/3}$ and $\pi_0-\theta$, valid for $\lambda>0$ sufficiently
small and~$\pi_0 \in [\underline{\theta},\overline{\theta}]$ inside the
no-trade region. Its first terms are identical to~\eqref{eq:u expans};
$\pi_0-\theta$ only appears from the $O(\lambda^{4/3})$-term onwards.

\section{Solution of the dual problem}
We conclude by showing that the shadow price $\tilde{S}$ and the density process of the corresponding martingale measure are a solution of the dual problem from Definition~\ref{def:dual}.

\begin{theorem}\label{thm:dual}
Let $Y$ and $\tilde{S}$ be defined as Lemma \ref{thm:shadow}, and set

\[
  Z=\scr{E}\left(\int_0^\cdot -\frac{\tilde{\mu}(Y_t)}{\tilde{\sigma}(Y_t)}dW_t\right).
\]
Then 
$$u_{(1-\lambda)S,S}(\eta_B,\eta_S)=v_{(1-\lambda)S,S}(\eta_B,\eta_S),$$
i.e., there is no duality gap in the market with bid-ask process $((1-\lambda)S,S)$, and $(\tilde{S},Z)$ is a dual optimizer.
\end{theorem} 

\begin{proof}
First notice that for any continuous semimartingale $\bar{S}$ with values in the bid-ask spread $[(1-\lambda)S,S]$, it follows as in the proof \cite[Theorem 4.6]{kallsen.muhlekarbe.10} that
$$u_{(1-\lambda)S,S}(\eta_B,\eta_S) \leq u_{\bar{S},\bar{S}}(\eta_B,\eta_S).$$
Together with, e.g., \cite[Lemma 2.3]{kallsen.goll.00}, this implies
$$u_{(1-\lambda)S,S}(\eta_B,\eta_S) \leq v_{(1-\lambda)S,S}(\eta_B,\eta_S).$$
For the shadow price $\tilde{S}$, we have $v_{(1-\lambda)S,S}(\eta_B,\eta_S) \leq v_{\tilde{S},\tilde{S}}(\eta_B,\eta_S)$ by definition and $u_{\tilde{S},\tilde{S}}(\eta_B,\eta_S)=u_{(1-\lambda)S,S}(\eta_B,\eta_S)$ by Theorem \ref{thm:shadow2}.  It therefore remains to show
$$v_{\tilde{S},\tilde{S}}(\eta_B,\eta_S)=u_{\tilde{S},\tilde{S}}(\eta_B,\eta_S),$$
i.e., that there is no duality gap in the frictionless market with price process $\tilde{S}$. As the frictionless market with price process $\tilde{S}$ is standard and complete in the sense of \cite[Definition 1.7.3]{karatzas.shreve.98}, this follows from \cite[Theorem 3.9.12]{karatzas.shreve.98} and we are done.
\end{proof}

\bibliographystyle{acm}
\bibliography{functionalshadow}

\end{document}